\documentclass[superscriptaddress]{revtex4}
\usepackage{amsmath}
\usepackage{amsthm}
\usepackage{amsfonts}
\usepackage{graphicx}
\usepackage{bbold}

\newtheorem{thm}{Theorem}
\newtheorem{lem}{Lemma}

\begin{document}

\title{Biomimetic Cloning of Quantum Observables}

\date{\today}

\author{U. Alvarez-Rodriguez} \email{unaialvarezr@gmail.com}
\affiliation{Department of Physical Chemistry, University of the Basque Country UPV/EHU, Apartado 644, 48080 Bilbao, Spain} 
\author{M. Sanz}
\affiliation{Department of Physical Chemistry, University of the Basque Country UPV/EHU, Apartado 644, 48080 Bilbao, Spain}
\author{L. Lamata}
\affiliation{Department of Physical Chemistry, University of the Basque Country UPV/EHU, Apartado 644, 48080 Bilbao, Spain}
\author{E. Solano}
\affiliation{Department of Physical Chemistry, University of the Basque Country UPV/EHU, Apartado 644, 48080 Bilbao, Spain}
\affiliation{IKERBASQUE, Basque Foundation for Science, Alameda Urquijo 36, 48011 Bilbao, Spain}

\begin{abstract}
We propose a bio-inspired sequential quantum protocol for the cloning and preservation of the statistics associated to quantum observables of a given system. It combines the cloning of a set of commuting observables, permitted by the no-cloning and no-broadcasting theorems, with a controllable propagation of the initial state coherences to the subsequent generations. The protocol mimics the scenario in which an individual in an unknown quantum state copies and propagates its quantum information into an environment of blank qubits. Finally, we propose a realistic experimental implementation of this protocol in trapped ions.
  \end{abstract}

\maketitle

Quantum information is a research field that studies how to perform computational tasks with physical platforms in the quantum regime. Coping with complex quantum systems could give rise to an exponential gain in computational power and a new branch of possibilities as compared with classical computing~\cite{Feynman82,Lloyd96}. One of the turning points of quantum information is the no-cloning theorem~\cite{nc}, which expresses the impossibility of copying an unknown state. Therefore, the notion of perfect universal quantum cloning was abandoned, and replaced by the cloning of restricted families of states or cloning with imperfect fidelities. A paradigmatic instance is the Buzek and Hillery universal quantum cloning machine~\cite{bh}, among other cases~\cite{ska,lk,gams,luk,med}. Another approach is {\it partial quantum cloning}, consisting in the copy of partial quantum information of a quantum state. In this sense, an interesting example is the cloning of the statistics associated with an observable~\cite{Paris}. However, these methods are limited due to the classical character of the information one replicates, since it is impossible to clone two non-conmuting observables with the same unitary~\cite{nld,joint}.

For a long time, human beings mimicked nature to create or optimize devices and machines, as well as industrial processes and strategies. In particular, biomimetics is the branch of science which designs materials and machines inspired in the structure and function of biological systems~\cite{rev, vie, qb, cil, sch, qal}. Analogously, novel quantum protocols may be envisioned by mimicking macroscopic biological behaviors at the microscopic level, in what we call {\it quantum biomimetics}. 

Living entities are characterized by features such as self-reproduction, mutations, evolution or natural selection. Among them, the ability to self replicate is the most basic one. In fact, even though they are allowed to perfectly replicate classical information, biological systems only reproduce part of this information in the following generations. A paradigmatic example is DNA replication inside the nucleus of the cell, since only sequences of bases are copied, but not all of the information about the physical state of the molecule.

\section*{Results}
In this article, we propose a formalism for cloning partial quantum information beyond the restrictions imposed by the aforementioned no-go theorems. We use a family of increasingly growing entangled states~\cite{got,ent}  in order to preserve and propagate the information of an initial state. In particular, we are able to transmit more than just classical information to the forthcoming generation, i.e., both the diagonal elements and coherences of the associated density matrix. Finally, we analyze the feasibility of a possible experimental implementation with trapped ions. 
 
To introduce our protocol, let $\rho\in \mathfrak{B} (\mathbb{C},n)$ and $\rho_e \in \mathfrak{B}(\mathbb{C},n)$ be an arbitrary state and a blank state, respectively, and let $\theta$ be a Hermitian operator.  We define the cloning operation $U(\theta,\rho_e)$ as
\begin{equation}
\label{blugo}
\langle \theta \rangle _{\rho}\equiv\textrm{Tr}(\rho\theta)=\langle \theta \otimes \mathbb{1} \rangle _{U(\rho \otimes \rho_e)U^{\dagger}} = \langle \mathbb{1} \otimes \theta \rangle _{U(\rho \otimes \rho_e)U^{\dagger}}.
\end{equation}
We denote each subspace as an {\it individual}, see Fig.~\ref{dos}. The expectation value of $\theta$ in the initial state is cloned into both subspaces of the final state. This is the cloning machine for observables introduced in Ref.~\cite{Paris}. Here, we extend their results to an arbitrary dimension and show the existence of an additional operator $\tau$, which does not commute with $\theta$, and whose statistics is encoded in the global state of the system
\begin{equation}
\label{gringo}
 \langle \tau \rangle_{\rho}\equiv\textrm{Tr}(\rho\tau)= \langle \tau \otimes \tau \rangle_{U(\rho \otimes \rho_e)U^{\dagger}}.
\end{equation}

\begin{figure}[h!!!]
\begin{center}
\includegraphics[width=6.8 cm]{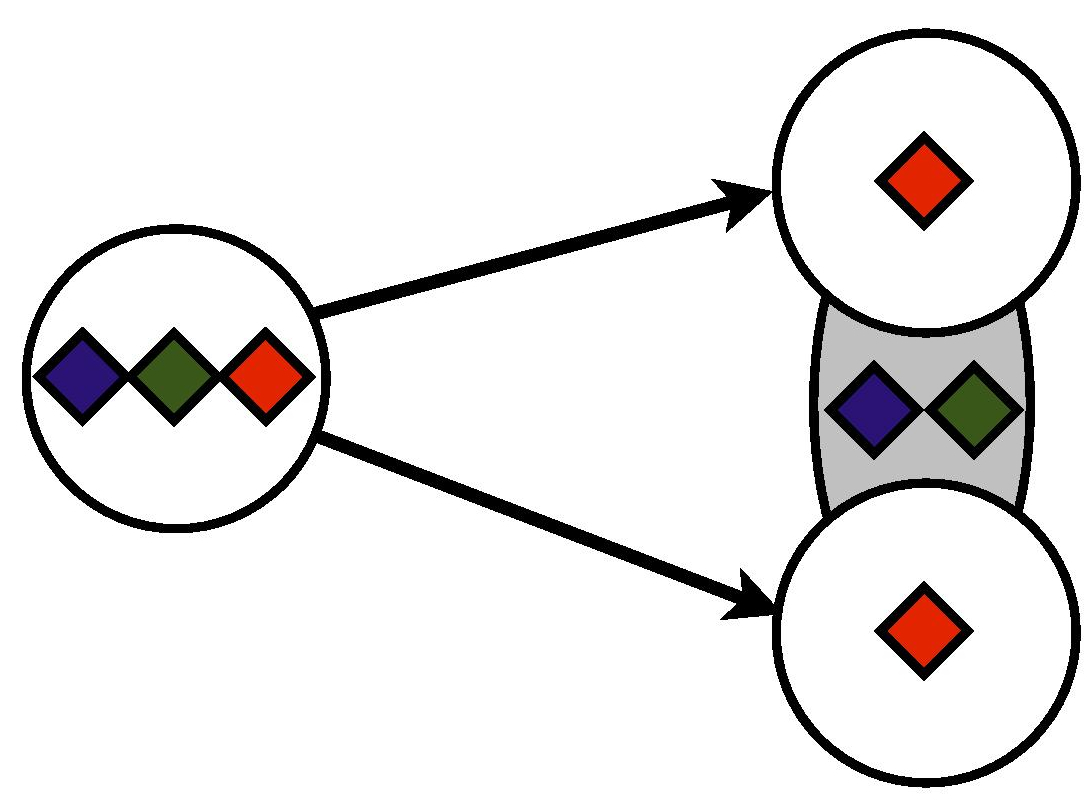}
\caption{\textbf{Cloning and transmission of quantum information.} In this scheme of our protocol the individuals are plotted with circles. The red diamond (centered inside each circle) represents the information that is cloned in every individual of the forthcoming generation, shown in Eq.~\eqref{blugo}. The other diamonds represent the information that is transmitted onto the global state of each generation relative to Eq. \eqref{gringo}.}
\label{dos}
\end{center}
\end{figure}

At this stage, it is convenient to point out the differences between both processes. On the one hand, Eq.~\eqref{blugo} can be regarded as a replication of a characteristic from an individual into its progeny. In this sense, the global state is traced out and the only important feature is the information of the individual. On the other hand, Eq.~\eqref{gringo} can be considered as a spread of the quantum coherences of the initial individuals to the forthcoming generations as a whole. Our formalism combines the notions of cloning and preserving quantum information without contradicting the no-cloning and no-broadcasting theorems~\cite{nbt}. This allows the propagation of the statistics of two non-commuting observables in a controllable way.

The explanation of the cloning method requires a selection of the basis, provided here by an external environment. The dark state of an unknown environment dynamics is the blank qubit for the copying process, i.e., the state that we define as the $|0\rangle$. This point will be relevant later in the discussion about the quantumness of the process.

\textbf{Cloning method.}
Let us work without loss of generality in a basis in which $\theta$ is diagonal. Then, we define the cloning operation $U_{n}(\theta, \rho_e)$ in terms of the generator of $n$-dimensional irreducible representation of the translation group $\{ x_{ni} \}$ and the projectors into each subspace $\{s_{ni} \}$. We clarify that $U_n$, after all, does not explicitly depend on $\theta$.
\begin{eqnarray}
x_{ni} |k\rangle &=&\left \{ \begin{array}{lcr} |k+i-1\rangle & \textrm{if} & k\le n-i+1 ,\\ |k+i-1-n\rangle & \textrm{if} & k >  n-i+1 ,\end{array} \right. \\  s_{ni}&=& |i\rangle\langle i | ,\\ \label{tpm} U_n&=& \sum^{n}_{i=1} s_{ni} \otimes x_{ni} = \bigoplus^{n}_{i=1} x_{ni}.
\end{eqnarray}
For example, for $n=2$, $U_2=U_{\textrm{CNOT}}$, and for $n=3$, $U_3$ is given by
\begin{equation}
\label{mono}
 U_3 = \left( \begin{array}{ccc} 1&0&0\\0&1&0\\0&0&1 \end{array} \right)\oplus \left( \begin{array}{ccc} 0&0&1\\1&0&0\\0&1&0 \end{array} \right)\oplus \left( \begin{array}{ccc} 0&1&0\\0&0&1\\1&0&0 \end{array} \right).
\end{equation}

We will demonstrate below that Eq.~\eqref{tpm} fullfils Eq.~\eqref{blugo}. The forthcoming theorems are proved in the Supplementary Information~\cite{txup}.

\begin{thm}
Let $\mathcal{H}\in\mathbb{C}^n$ be a Hilbert space of dimension $n$, $U\in\mathcal{H}\otimes\mathcal{H}$ the unitary operation defined in Eq.~\eqref{tpm}, $\rho_e=|0\rangle\langle0|$, and $\rho,\theta$ $\in\mathfrak{B}(\mathcal{H})$ bounded Hermitian operators. Then, the unitary $U$ satisfies Eq.~\eqref{blugo}.
\end{thm}

We analyze now if the cloning unitary $U$ transmits information from the initial individual to the progeny as a whole, apart from cloning its expectation value independently. The process mimics the information transmission underlying collective behaviors present in many biological systems, such as self-organizing neurons.

\begin{thm} 
Let $U\in\mathcal{H}\otimes\mathcal{H}$ be the unitary transformation defined in Eq.~\eqref{tpm}. Then, there exists a bounded antidiagonal operator $\tau$, whose matrix elements are $0$ or $1$, fulfilling Eq.~\eqref{gringo}.
\end{thm}

Notice that, in previous theorems, we have worked in a basis in which $\theta$ is diagonal. However, the cloning operation can be rewritten in any basis just by rotating the matrix with the proper unitary, $\theta'=R^\dag \theta R$, transforming $\tau$ into  $\tau'=R^\dag \tau R$ and $U$ into $U'=(R^\dag \otimes R^\dag)U(R\otimes\mathbb{1})$.
 
When the cloning operation is sequentially reapplied, it propagates the information of the initial state, i.e., it transmits the full statistics of the density matrix, Tr$[\rho\sigma_{i}]$, $i=1,2,3$. We show here how two-qubit operations are extended into $m$-qubit states,
\begin{eqnarray}
\label{tomate}
\nonumber U_{i,i+\frac{m}{2}}= x_1^{\otimes i-1}\otimes s_1 \otimes x_1^{\otimes m-i}+ x_1^{\otimes i-1}\otimes s_2 \otimes x_1^{\otimes \frac{m}{2}-1}\otimes x_2 \otimes x_1^{\otimes \frac{m}{2}-i} ,
\end{eqnarray}
where the subscripts of $U$ refer to the pair of qubits that is acted upon. The cloning for the subsequent generations, see Fig.~\ref{mas}, is constructed through the product of pairwise cloning operations, $U=\prod^{\frac{m}{2}}_{i=1} U_{i,i+\frac{m}{2}}$. For instance, the density matrices of the first and second generations reads
\begin{eqnarray}
\rho_{1}&=&(U_{1,2}) (\rho_0 \otimes \rho_{e}) (U_{1,2})^{\dag}, \\ \label{nueve} \rho_{2}&=&(U_{1,3}\,U_{2,4})(\rho_{1}\otimes\rho_{e}\otimes\rho_{e})(U_{1,3}\,U_{2,4})^{\dag} .
\end{eqnarray}
Therefore, the mechanism is straightforwardly generalizable for obtaining sequential generations with the same information in each individual, in the spirit of a quantum genetic code. Moreover, although Eq.~\eqref{tomate} holds only for qubits, it is noteworthy to mention that an extension to higher dimensions is possible. This can be achieved using as building block the unitary gate $U_n$ defined in Eq.~\eqref{tpm}.

\begin{figure}[h!]
\begin{center}
\includegraphics[width=7.5cm]{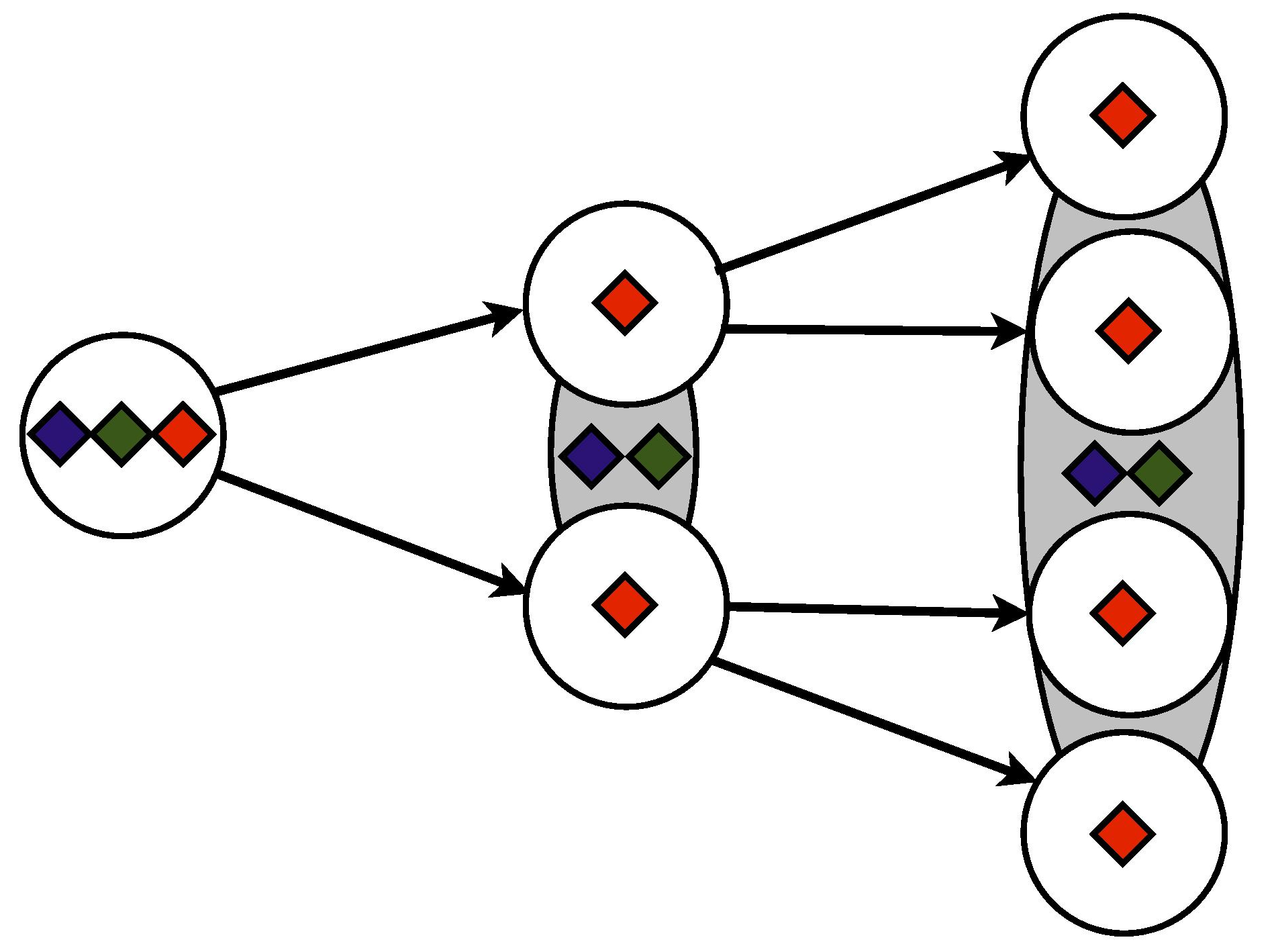}
\caption{\textbf{Iteration of the cloning process.} Scheme of the sequential cloning of the information encoded in an initial state into individuals of subsequent generations.}
\label{mas}
\end{center}
\end{figure}

We will present now a counter-example showing that the cloning operation is not unique. When $n$ is not a prime number, $n=kl$,  there are other cloning unitaries apart from $U_{kl}$. For instance, an additional $U'_{kl}$ is constructed via the solutions in each subspace, $U_k$ and $U_l$, respectively,
\begin{equation}
\label{wek}
U'_{kl}=\sum^{k}_{i=1}\sum^{l}_{j=1}s_{ki}\otimes s_{lj} \otimes x_{ki} \otimes x_{lj}.
\end{equation}
This result shows that it is possible to mix information among subspaces of different dimensions. An example for $n=6$ is depicted in Fig.~\ref{uno}.\\

\begin{figure}[t!]
\begin{center}
\includegraphics[width=8.75cm]{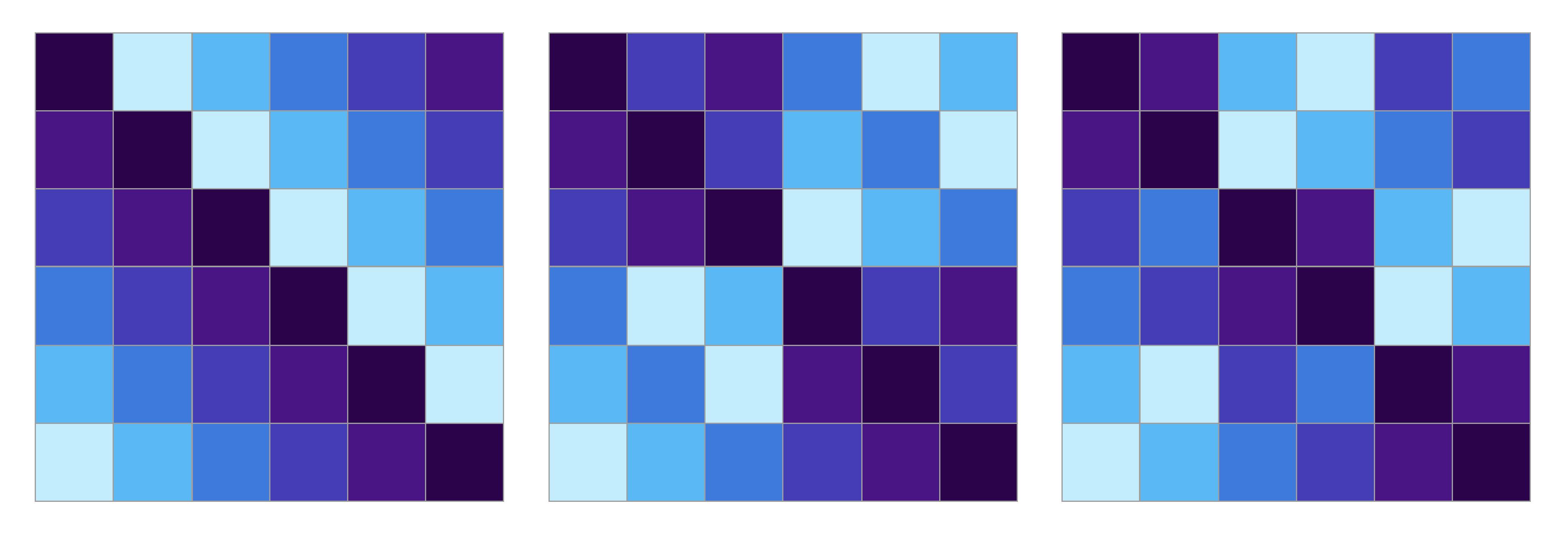}
\caption{\textbf{Cloning unitary operations for ${\bf n=6}$.} In these graphics, each color represents a matrix of the translation group. Therefore, each of the three two-dimensional arrays groups the six matrices of the translation group for $n=6$. The first array combines the $x_{ni}$ operations of Eq.~\eqref{tpm}, while the second and third show $x_{ki} \otimes x_{lj}$ of Eq.~\eqref{wek} with $k,l=2,3$ and $k,l=3,2$, respectively.}
\label{uno}
\end{center}
\end{figure}

\textbf{Quantumness and Classicality.}
In this section, we analyze the quantumness of the proposed biomimetic cloning protocol. In Ref.~\cite{cri}, Meznaric {\it et al.} recently introduced a criterion to determine the nonclassicality of an operation $\Omega$. This method is based on the distance between the outcome of the operation and the pointer basis einselected by the environment. The quantum operation $\Omega$ is composed with the completely dephasing channel $\Gamma$, provided by the environment. The measurement of nonclassicality is obtained by maximizing over all states the relative entropy of both operations acting on an arbitrary quantum state. The completely dephasing channel plays the role of an external environment that einselects the pointer basis, in which $\Omega$ is represented. Effectively, as proposed in Ref.~\cite{cri}, any operation is quantum whenever its column vectors are superpositions of the elements of the basis. On the contrary, the operation is classical if they are just permutations of the basis.

It is natural to identify the basis of the ancillary qudits, $\rho_e$, with the pointer states, since we assume that the system naturally provides blank qudits. By applying the classification formalism for qubits, the process of copying $\sigma_{z}$ in the individual state and $\sigma_{x}$ in the global state is classical, because the $U_{\textrm{CNOT}}$ is written in terms of permutations of the pointer states. Nevertheless, the complementary operation 
\begin{eqnarray}
U_x=\frac{1}{\sqrt{2}} \left( \begin{array}{cccc} 1&1&0&0\\0&0&1&-1\\0&0&1&1\\1&-1&0&0 \end{array}\right), 
\end{eqnarray}
which clones $\sigma_{x}$ in the individual state and $\sigma_{z}$ in the global state, is quantum.

Another possibility is to consider that the classical pointer basis is defined by the $\Omega$ operator itself. This means that we can construct the quantum channel $\mathcal{E} (\rho)$ which maps the initial state to any of the outcomes, considering the blank qudit and the other outcome as the environment. By construction, the unitary given by Eq.~\eqref{tpm} leads to an injective channel, since the number of Kraus operators is the same as the square of the matrix dimension. Therefore, the only fixed point is the identity and the cloning operations are classical when written in the einselected basis~\cite{mik}.

According to these results, the cloning formalism copies classical information but preserves quantum correlations, which makes the global operation quantum. The interpretation of this property is that the quantum part of the information is encoded  in the global state. By analogy with biological systems, the environment plays a fundamental role in the kind of information that, similar to quantum darwinism~\cite{qdw}, is preserved and cloned through a pure quantum mechanism. In our work, the quantumness is only revealed when considering collectively all outcomes of the copying process. 

\textbf{ Experimental Implementation.}
We consider that an experimental realization of our protocol in a quantum platform sets a significant step towards quantum artificial life. We propose an experimental setup of our formalism in trapped ions, arguably the most advanced quantum technologies in terms of coherence times and gate fidelities~\cite{boat}. Current experimental resources would allow copying processes for qubits and qutrits of three and two generations, respectively. The number of logical gates and trapped ions required to perform the experiment is presented in Table~\ref{reo}. We point out that the only limitations for performing trapped-ion experiments with higher dimensions and larger number of generations are decoherence times and gate errors. None of these are fundamental and near future improvements may allow to reach the implementation of higher dimensional individuals and many more generations with  near perfect fidelities.

\begin{table}[h!]
\caption{{\bf Technological and computational resources.} We show the number of quantum gates and trapped ions needed to perform the cloning experiment with qubits and qutrits, respectively.\\}
\centering
\begin{tabular}{|lcccc|}
\hline
\multicolumn{5}{| c |}{Qubit}\\
\hline
Sequential generation step &  0 & 1 & 2 & 3 \\ \hline
Total gates (2 qubit gates) & 0 & 19(2) & 57(6) & 133(14) \\ \hline
Ions & 1 & 2 & 4 & 8 \\ 
\hline\hline
\multicolumn{5}{| c |}{Qutrit}\\
\hline
Sequential generation step &  0 & 1 & 2 & - \\ \hline
Total gates (2 qubit gates) & 0 & 38(4) & 114(12) & -  \\ \hline
Ions & 2 & 4 & 8 & - \\ 
\hline 
\end{tabular}
\label{reo}
\end{table}

We encode the initial state in one of the ions, while the rest are initialized in the $|0\rangle$ state. The cloning operation for two qubits is the CNOT gate, which can be reproduced performing the M$\o$lmer-S$\o$rensen gate~\cite{ms} and a sequence of single qubit gates~\cite{cz}, $U_{\textrm{CNOT}}=-(\sigma_{x}\otimes\sigma_{z}) P_{1}P_2^{-1}H_{2}RP_{1}H_{1}P_{1}RP_2$. Here, $P$ is the phase gate, $H$ is the Hadamard gate, $R$ is the M$\o$lmer-S$\o$rensen gate, and the subindex denotes the ion number. We express the gates as products of carrier transitions $R^c$ and a phase factor,
\begin{eqnarray}
P=e^{-i3\pi/4}R^c (\pi,0) R^c (\pi,\pi/4),\\  H=e^{-i\pi/2} R^c(\pi,0)R^c(\pi/2,\pi/2), \\ \sigma_{z}=e^{-i\pi/2} R^c(\pi, 0)R^c(\pi,\pi/2).
\end{eqnarray}
The first and second generations of cloned qubits are obtained as indicated in Eq.~\eqref{nueve}. The next step in the copying process, i.e., the third generation, is given by $\rho_{3}=(U_{1,5} \, U_{2,6} \, U_{3,7} \, U_{4,8}) (\rho_{2} \otimes \rho_{e}^{\otimes 4} )(U_{1,5} \, U_{2,6} \, U_{3,7} \, U_{4,8})^{\dag}$.\\

We consider now the implementation of the qutrit case in trapped ions. As there is no direct access to qutrit gates, we have to engineer a protocol with one and two-qubit gates. We suggest to add three ancillary levels that split the unitary operation into three subspaces of $4 \times 4$ matrices. To achieve this, for the $U_{3}$ given in Eq.\eqref{mono}, the modified $U'_3$ is
\begin{equation}
 U'_{3}= \mathbb{1}_{4} \oplus \left( \begin{array}{cccc}  1&0&0&0\\0&0&0&1\\0&1&0&0\\0&0&1&0 \end{array} \right) \oplus \left( \begin{array}{cccc}  1&0&0&0\\0&0&1&0\\0&0&0&1\\0&1&0&0 \end{array} \right).
\end{equation} 
The first submatrix does not require any quantum gate. The second and third submatrices are the products of two CNOT gates, for which the role of control and target ions is interchanged. Hence, the implementation is reduced from qutrit to qubit operations.

\section*{Discussion}
In summary, inspired by biological systems, we have brought concepts and applications into quantum information theory. For instance, our partial quantum cloning method makes use of global and local measurements in order to encode information of nonconmuting observables beyond the classical realm. Moreover, we prove that the information transmission is purely quantum for a certain kind of operators. In parallel, we show that it is possible to implement our ideas in an ion-trap platform with current technology. 

Replication is the most fundamental property that one may require from a biological system. We leave to forthcoming works the introduction of additional biological behaviors such as mutations, evolution and natural selection, englobed by the frame of quantum artificial life. This proposal should be considered as the first step towards mimicking biological behaviours in controllable quantum systems, a concept that we have called quantum biomimetics.

\section*{Acknogledgements}
The authors acknowledge funding from Basque Government  BFI-2012-322 and IT472-10 grants, Spanish MINECO FIS2012-36673-C03-02, Ram\'on y Cajal Grant RYC-2012-11391, UPV/EHU UFI 11/55, CCQED, PROMISCE, and SCALEQIT European projects.

\section*{Author Contributions}
U. A.-R. made the calculations while U. A.-R., M. S., L. L. and E. S. developed the model and wrote the manuscript.

\section*{Additional Information}
{\bf Competing financial interests:} The authors declare no competing financial interests.\\

\pagebreak
\center{{\LARGE Supplementary Material}}\\
\vspace{1cm}
In this Supplemental Material we add mathematical calculations that complement the information of the article. 

\section*{Theorems}
$\bullet$ Here we give the proof of theorem 1 in the paper.
\begin{proof}
\begin{eqnarray*}
\langle \theta\otimes\mathbb{1}\rangle &=& \textrm{Tr}[U (\rho\otimes\rho_{e})U^{\dag} (\theta\otimes\mathbb{1})] \\ &=&\textrm{Tr} [\sum^{n}_{i,j=1} s_i \rho s^{\dag}_j \theta \otimes x_i \rho_{e} x^{\dag}_{j}].
\end{eqnarray*}
When looking at the second subspace, $x_i \rho_e x^\dag _j$, only the terms with $i=j$ are left in the diagonal, therefore only those are relevant for the trace. We consider $i=j$ in the first subspace, $s_i \rho s^\dag _i \theta$. The remaining is the product between the diagonal terms of the density matrix $\rho$ and the observable $\theta$, that is precisely.
\begin{equation}
\nonumber \textrm{Tr}[\rho \theta]=\langle \theta \rangle
\end{equation}
This shows that the mean value of $\theta$ is effectively cloned in the first individual by means of $U$. The proof for the second individual is obtained in a similar way.
\end{proof}

$\bullet$ We enunciate and prove a lemma that clarifies the properties of the cloning operation U.
\begin{lem}
Let $U: \mathcal{H}\otimes\mathcal{H}\to \mathcal{H}\otimes\mathcal{H}$ be the unitary transformation given by the previous theorem. Then, there exists a bounded operator $\tau \neq 0,1$ such that Eq. $(2)$ is fulfilled iff $U = (V \otimes V)(w_1 \oplus w_2) F (V \otimes V)^{\dagger}$, with $V$ and $w_1, w_2$ arbitrary unitary matrices of dimensions $d$ and $\frac{d^2}{2}$ respectively, and $F$ a flip operator such that $F (\sigma_z \otimes 1_{\frac{d}{2}} \otimes \sigma_z \otimes 1_{\frac{d}{2}}) F^{\dagger} = \sigma_z \otimes 1_{\frac{d^2}{2}} $.
\end{lem}

\begin{proof}
Let us assume that $\tau$ is hermitian, since otherwise $\tau^{\dagger}$ must also be a solution. Then, we can work without lost of generality in the basis in which $\tau$ is diagonal, i.e. $\tau = V D V^{\dagger}$, $D$ a diagonal matrix. Therefore, Eq. $(2)$ is transformed into the equation $\tilde{U} (D\otimes D) \tilde{U}^{\dagger} = D \otimes 1_d$, with $\tilde{U} = (V^{\dagger} \otimes V^{\dagger}) U (V\otimes V)$ also unitary. The necessary condition for the existence of $\tau$ is $\mathit{spec}(D\otimes D) = \mathit{spec}(D \otimes 1_d)$, which trivially means that the eigenvalues of $D$ must be $0$, $1$ or $-1$, with the same degeneration on both sides. If we call $\lambda_{-1}$, $\lambda_0$ and $\lambda_1$ to the degeneration of the eigenvalues $-1$, $0$ and $1$ respectively, a simple counting with the condition $\lambda_{1} +\lambda_{0}+\lambda_{-1} = d$ leads to the equations
\begin{align*}
d \lambda_1 &= \lambda_{1}^{2} + \lambda_{-1}^{2} \\
d \lambda_{0} &= \lambda_{0} (2d-\lambda_{0}) \\
d \lambda_{-1} &= 2\lambda_{1} \lambda_{-1}
\end{align*}
whose only non--trivial solution is $\lambda_{1} = \lambda_{-1} = \frac{d}{2}$ and $\lambda_{0} = 0$ when $d$ is \textit{even}. In other words, $\tau$ must be a unitary rotation of $D = \sigma_z \otimes 1_{(\frac{d}{2})}$. This provides the necessary condition for $(\tau \otimes \tau)$ and $(\tau \otimes 1_d)$ to be connected by \textit{a} unitary but not for the case in which the unitary is already fixed. So let us determine for which unitary matrices such a $\tau$ exists. 

Let $F$ be a flip operator which transforms $F (\sigma_z \otimes 1_{\frac{d}{2}} \otimes \sigma_z \otimes 1_{\frac{d}{2}}) F^{\dagger} = \sigma_z \otimes 1_{\frac{d^2}{2}} $. Then, we have $W \sigma_z \otimes 1_{\frac{d^2}{2}} W^{\dagger} = \sigma_z \otimes 1_{\frac{d^2}{2}} $, with $W = \tilde{U} F^{\dagger}$. This means that the most general possible $W = w_1 \oplus w_2$, with $w_1$, $w_2$ two unitary matrices of dimension $\frac{d^2}{2}$ and the result follows straightforwardly.
\end{proof}

$\bullet$ Now we prove theorem 2 in the paper.
\begin{proof}
\begin{eqnarray}
\langle \tau\otimes\tau \rangle&=&\textrm{Tr}[U(\rho\otimes\rho_{e})U^{\dag}(\tau\otimes \tau)]\\ \nonumber &=&\textrm{Tr}[\sum^{n}_{i,j=1}s_{i}\rho s^{\dag}_{j} \tau \otimes x_{i} \rho_{e} x^{\dag}_{j} \tau].
\end{eqnarray}
The second subspace, $x_i \rho x^\dag _j \tau$, will only contribute to the trace when it is diagonal. Since $\tau$ is antidiagonal the only effective combination of $i,j$ is ($i=1, j=n$), ($i=2,j=n-1$),$\ldots$ When applying this restriction to the first subspace, $s_i \rho s^\dag _j \tau$, the global expresion turns out to be Tr[$\rho \tau$], only if the matrix elements of $\tau$ are $0$ or $1$. We have thus shown how $\tau$ is promoted into the global state of the system.\\
\end{proof}

\section*{Properties of the cloning operation}
$\bullet$ We show here that the $\tau$ operator is also transmitted when cloning in any basis. When the unitary operation $U$ is rotated with the matrix $R$ it propagates the information of the rotated matrix $\tau '$. 
\begin{eqnarray*}
&&\text{Tr}[\rho \tau']=\text{Tr}[U' (\rho \otimes \rho_e) U'^\dag ( \tau\otimes\tau ) ]= \\ && \text{Tr}[(R^{\dag} \otimes R^{\dag}) U (R \otimes \mathbb{1})( \rho \otimes \rho_{e})( R^{\dag} \otimes \mathbb{1}) U^{\dag} \\ && ( R \otimes R)( R^\dag \tau R\otimes R^{\dag}\tau R)]=\\ && \text{Tr}[U (R\otimes \mathbb{1})( \rho\otimes\rho_e)( R^\dag \otimes \mathbb{1}) U^\dag( R\otimes R)( R^\dag \tau\otimes R^\dag \tau)]= \\ && \text{Tr}[U R\rho R^{\dag} U^{\dag} (\tau\otimes \tau)]=\text{Tr}[U \tilde{\rho} U^{\dag} (\tau \otimes \tau)]=\text{Tr}[\tilde{\rho}\tau].
\end{eqnarray*}
\begin{eqnarray*}
\text{Tr}[\rho \tau']=\text{Tr}[\rho R^{\dag}\tau R]=\text{Tr}[R \rho R^{\dag} \tau]=\text{Tr}[\tilde{\rho}\tau].
\end{eqnarray*}

\end{document}